\documentclass[12pt, a4paper]{article}

\expandafter\let\csname equation*\endcsname\relax
\expandafter\let\csname endequation*\endcsname\relax
\usepackage{amsmath}
\usepackage{amssymb}
\usepackage{mathrsfs}
\usepackage[small,bf]{caption} 
\usepackage[subrefformat=parens]{subcaption}
\usepackage{here}

\usepackage{appendix}

\usepackage{amsthm}
\theoremstyle{plain}

\newtheorem*{thm*}{Theorem}
\newtheorem*{defi*}{Definition}
\newtheorem*{lem*}{Lemma}
\newtheorem*{cor*}{Corollary}
\newtheorem*{prop*}{Proposition}

\newcommand{\R}{\mathbb{R}}

\usepackage{authblk}

\usepackage[]{graphicx}

\newcommand{\kB}{k_{\mathrm{B}}}

\begin{document}

\title{Entropy of mixing exists only for classical and quantum-like theories among the regular polygon theories}

\author{Ryo Takakura\thanks{takakura.ryo.27v@st.kyoto-u.ac.jp}}
\affil{Department of Nuclear Engineering 
	\\Kyoto University
	\\Kyoto daigaku-katsura, Nishikyo-ku, Kyoto, 615-8540, Japan}
\date{}

\maketitle
\begin{abstract}
	The thermodynamical entropy of a system which consists of different kinds of ideal gases is known to be defined successfully in the case when the differences are described by classical or quantum theory. Since these theories are special examples in the framework of generalized probabilistic theories (GPTs), it is natural to generalize the notion of thermodynamical entropy to systems where the internal degrees of particles are described by other possible theories. In this paper, we consider thermodynamical entropy of mixing in a specific series of theories of GPTs called the regular polygon theories, which can be regarded from a geometrical perspective as intermediate theories between a classical trit and a quantum bit with real coefficients. We prove that the operationally natural thermodynamical entropy of mixing does not exist in those inbetween theories, that is, the existence of the natural entropy results in classical and quantum-like theories among the regular polygon theories.
\end{abstract}

\section{Introduction}
The concept of entropy plays an important role in thermodynamics \cite{Callen_thermo,zemansky_thermo}. It is possible to calculate the thermodynamical entropy of a mixture of classically different kinds of particles (such like a mixture of nitrogens and oxygens), and similar ideas were applied by von Neumann to the case when the system was composed of particles with different quantum internal states \cite{von1955mathematical}. On the other hand, both classical and quantum theories can be classified as special cases of {\it generalized probabilistic theories} (GPTs), which have an advantage in that they describe operationally and intuitively physical experiments and can be considered as the most general framework of physics \cite{Gudder_stochastic,Araki_mathematical,hardy2001quantum,PhysRevA.75.032304,BARNUM20113,PhysRevA.84.012311,1751-8121-47-32-323001}. There have been researches which aim to introduce and investigate the concept of entropy in GPTs from informational perspectives \cite{1367-2630-12-3-033024,Short_2010,KIMURA2010175,KimuraEntropiesinGeneralProbabilisticTheoriesandTheirApplicationtotheHolevoBound}. In those researches, some kinds of entropy were defined in all theories of GPTs and their information-theoretical properties were investigated. Meanwhile, there have been also researches referring to the thermodynamical entropy in terms of the microcanonical or canonical formulation in GPTs \cite{Chiribella_2017,chiribella2016entanglement}, and researches referring to the thermodynamical entropy of mixing in GPTs \cite{1367-2630-19-4-043025,EPTCS195.4}. However, in those work, the entropy was only defined in or applied to some restricted theories of GPTs with special assumptions.

In GPTs, a state of the classical trit system is described by an element of a triangle-shaped state space. On the other hand, a state of the simplest quantum system, the two-level system or qubit, is expressed by an element of the Bloch ball. In the study of GPTs, the ball is often substituted by a disk which is its two-dimensional counterpart \cite{doi:10.1063/1.4998711}. Since a disc-shaped state space is known to describe a qubit with real coefficients, or a ``quantum-like" bit, one can introduce theories whose state spaces are in the shapes of regular polygons as describing intermediate theories between a classical trit and a quantum-like bit in the framework of GPTs. In fact, those theories (the {\it regular polygon theories}) are known not to always satisfy the assumptions imposed in the previous studies. It seems natural to ask how entropy of mixing behaves in those theories.

In this paper, we consider thermodynamical entropy of mixing in the regular polygon theories. It is proved that the operationally natural thermodynamical entropy of a mixture of ideal particles with different internal states described by one of the regular polygon theories in GPTs exists if and only if the theory is either classical or quantum-like. More precisely, we demonstrate that the thermodynamical entropy of mixing satisfying conditions imposed in \cite{1367-2630-19-4-043025}, where the concrete operational construction of the entropy was given as von Neumann did under the assumption of the existence of semipermeable membranes, does not exist in all the regular polygon theories except for classical and quantum-like ones.

This paper is organized as follows. In section \ref{sec2}, we give a short review of GPTs, and introduce the regular polygon theories and thermodynamically natural entropy of mixing. We state our main theorem and its brief proof in section \ref{sec3}, and conclude this paper and note some future works in section \ref{sec4}.

\section{Fundamental concepts}
\label{sec2}

\subsection{GPTs}
\label{2a}
Here, we introduce briefly the mathematical framework of GPTs according mainly to \cite{1751-8121-47-32-323001,kimura2010physical}. In a physical experiment, we prepare a system being measured, measure a physical quantity, obtain one of the several outcomes (we only consider measurements with finite outcomes in this paper), and repeat this procedure to obtain statistics about the outcomes \cite{Araki_mathematical}. GPTs describe physical experiments in the most general way shown in the following.

In each theory of GPTs, preparation procedures are represented by {\it states}, and the set of all states of a GPT is called the {\it state space} of the theory. To describe the concept of probability mixtures, a state space $\Omega$ should be a convex set embedded in a vector space $V$, that is, if $\omega_{1}, \omega_{2}\in\Omega$, then for all $0\le p\le1$, $p\omega_{1}+(1-p)\omega_{2}\in\Omega$, which represents the mixture of two states $\omega_{1}$ and $\omega_{2}$ with probability weights $p, 1-p$. In this paper, we assume that every state space $\Omega$ is finite dimensional and closed, and embedded in $V=\R^{d}$. The extreme points of $\Omega$ are called ${\it pure states}$, and the other elements of $\Omega$ are called ${\it mixed states}$. 

To introduce the notion of measurements in GPTs, we define ${\it effects}$. Let $\Omega$ be the state space of some GPT. An effect $e$ is an affine function mapping a state $\omega\in\Omega$ into some value $e(\omega)\in[0, 1]$, which gives the probability of getting a specific outcome in the system prepared in $\omega$. Note that the affinity of effects ensures the concept of probability mixture of states. The set of all effects is called the ${\it effect\ space}$ and denoted by $E(\Omega)$, that is, $E(\Omega)=\{e\in V^{*}\mid \forall\omega\in\Omega,\ 0\le e(\omega)\le1\}$, where $V^{*}$ is the dual space of $V$. Note that we follow the {\it no-restriction hypothesis} in this paper \cite{PhysRevA.81.062348}. $E(\Omega)$ is also a convex set with its natural convex combinations in $V^{*}$, and there exists a special effect in $E(\Omega)$ called the ${\it unit\ effect}$. It is denoted by $u$, and satisfies $u(\omega)=1$ for all $\omega\in\Omega$ (in this paper, we do not consider unnormalized states, which are not mapped to 1 by $u$). It is easy to check $u$ is unique in $E(\Omega)$ and is an extreme effect of $E(\Omega)$, and if $e\in E(\Omega)$ then also $u-e\in E(\Omega)$ (moreover, if $e$ is an extreme effect, then $u-e$ is also an extreme effect). A ${\it measurement}$ (with $l$ outcomes) is defined by a set of effects $\{e_{1}, e_{2}, \cdots, e_{l}\}$ such that $e_{1}+e_{2}+\cdots+e_{l}=u$, where $e_{i}(\omega)$ represents the probability of getting the $i$th outcome in the system whose state is $\omega$ for each $i=1, 2, \cdots, l$.

A set of $m$ states $\{\omega_{1}, \omega_{2}, \cdots, \omega_{m}\}$ is called ${\it perfectly}$ ${\it distinguishable}$ if and only if there exists a measurement $\{e_{1}, e_{2}, \cdots, e_{m}\}$ such that $e_{i}(\omega_{j})=\delta_{ij},\ i,j=1,\ 2,\ \cdots,\ m$. In general, we can not identify the state of a system by a single measurement. However, for perfectly distinguishable states, there exists a measurement by which we can detect perfectly in which state the system is prepared.

\subsection{Regular polygon theories}
\label{2b}
Let $d$, the dimension of the vector space $V=\mathbb{R}^{d}$, equal to three. Then, both states and effects are represented by three-dimensional Euclidean vectors, and it is possible to take the affine actions of effects on states as Euclidean inner products of those vectors. We can introduce in accord with \cite{1367-2630-13-6-063024} the {\it regular polygon theories}, whose state spaces are in the shapes of regular polygons.

The regular polygon theories are composed of the $n$-gon theories. A GPT with its state space $\Omega$ in $\mathbb{R}^{3}$ is called the $n$-${\it gon\ theory}$ if and only if $\Omega$ is the regular polygon with $n$ sides, that is, $\Omega$ is the convex hull of the following $n$ pure(extreme) states
\begin{align*}
	\omega_{i}^{n}=
	\left(
	\begin{array}{c}
		r_{n}\cos({\frac{2\pi i}{n}})\\
		r_{n}\sin({\frac{2\pi i}{n}})\\
		1
	\end{array}
	\right),\ \ &r_{n}=\sqrt{\frac{1}{\cos({\frac{\pi}{n}})}}, \ \ i=0, 1, \cdots, n-1\ \ (n:\mbox{finite}).
\end{align*}
In the case when $n=\infty$, the state space is described by an unit disk, whose extreme points are 
\begin{align*}
	\omega_{\theta}^{\infty}=
	\left(
	\begin{array}{c}
		\cos\theta\\
		\sin\theta\\
		1
	\end{array}
	\right),\ \ \theta\in[0, 2\pi).
\end{align*}
We express the $n$-gon state space by $\Omega_{n}$, and do not consider the case when $n=1, 2$. The corresponding effect space $E(\Omega_{n})$ for $n=3, 4, \cdots, \infty$ is the convex hull of $0$, the unit effect $u={}^{t}(0, 0, 1)$, and the other extreme effects
\begin{align*}
	&e_{i}^{n}=\frac{1}{2}
	\left(
	\begin{array}{c}
		r_{n}\cos({\frac{(2i-1)\pi}{n}})\\
		r_{n}\sin({\frac{(2i-1)\pi}{n}})\\
		1
	\end{array}
	\right),\ \ i=0, 1, \cdots, n-1\ \ (n:\mbox{even})\ ;\\
	&e_{i}^{n}=\frac{1}{1+r_{n}^{2}}
	\left(
	\begin{array}{c}
		r_{n}\cos({\frac{2i\pi}{n}})\\
		r_{n}\sin({\frac{2i\pi}{n}})\\
		1
	\end{array}
	\right),\ \ \overline{e_{i}^{n}}\equiv u-e_{i}^{n},\ \ i=0, 1, \cdots, n-1\ \ (n:\mbox{odd})\ ;\\
	&e_{\theta}^{\infty}=\frac{1}{2}
	\left(
	\begin{array}{c}
		\cos\theta\\
		\sin\theta\\
		1
	\end{array}
	\right),\ \ \theta\in[0, 2\pi)\ \ \ (n=\infty).
\end{align*}
Remark that we have to consider the extreme effects $\{\overline{e^{n}_{i}}\}_{i=0}^{n-1}$ when $n$ is odd, since each $\overline{e^{n}_{i}}=u-e^{n}_{i}$ is an extreme effect as noted in the previous subsection and the set $\{\overline{e^{n}_{i}}\}_{i=0}^{n-1}$ does not coincide with the set $\{e^{n}_{i}\}_{i=0}^{n-1}$ for odd $n$
\footnote{It seems that the $n$-gon theory with even $n$ has a simpler structure of extreme effects than odd $n$, although the former is known to have weaker geometrical symmetry called {\it weak self-duality} between the state cone and the dual cone than the latter which have {\it (strong) self-duality} \cite{Barnum2008teleportation, Barnum2013}.}. The case when $n=3$ corresponds to the classical trit case, since any mixed state can be decomposed uniquely into the convex combination of the three pure states in that theory. By contrast, the case when $n=\infty$ can be considered as an analog of the qubit case \cite{doi:10.1063/1.4998711}. In fact, considering the usual Bloch representation of a qubit, we can see that the round state space represents some equatorial plane of the Bloch ball, the space of a qubit with real coefficients. Therefore, we call this theory quantum-like. 

Next, we characterize perfectly distinguishable states in the $n$-gon theory. 
\begin{figure}[h]
	\hfill
	\begin{minipage}[b]{0.3\linewidth}
		\centering
		\includegraphics[scale=0.4]{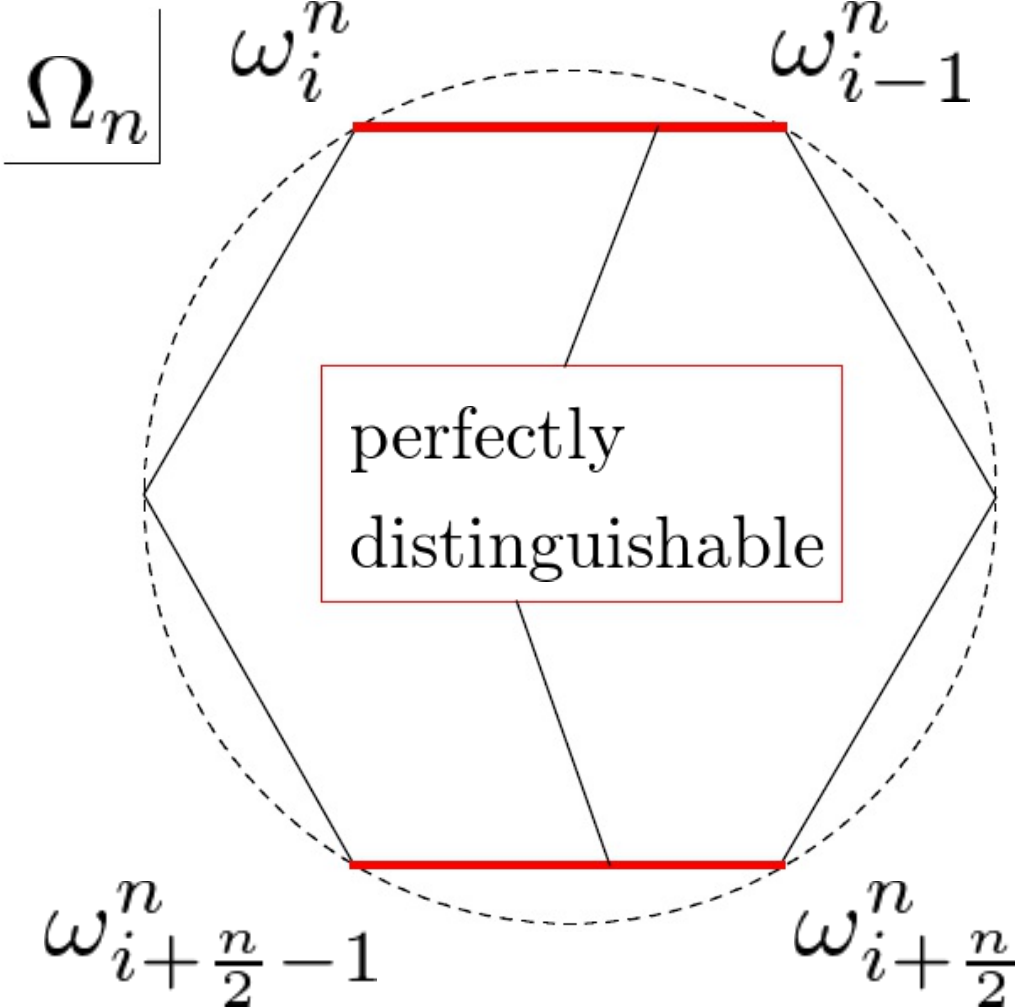}
		\subcaption{$n$ is an even number.}
	\end{minipage}
	\hfill
	\begin{minipage}[b]{0.3\linewidth}
		\centering
		\includegraphics[scale=0.4]{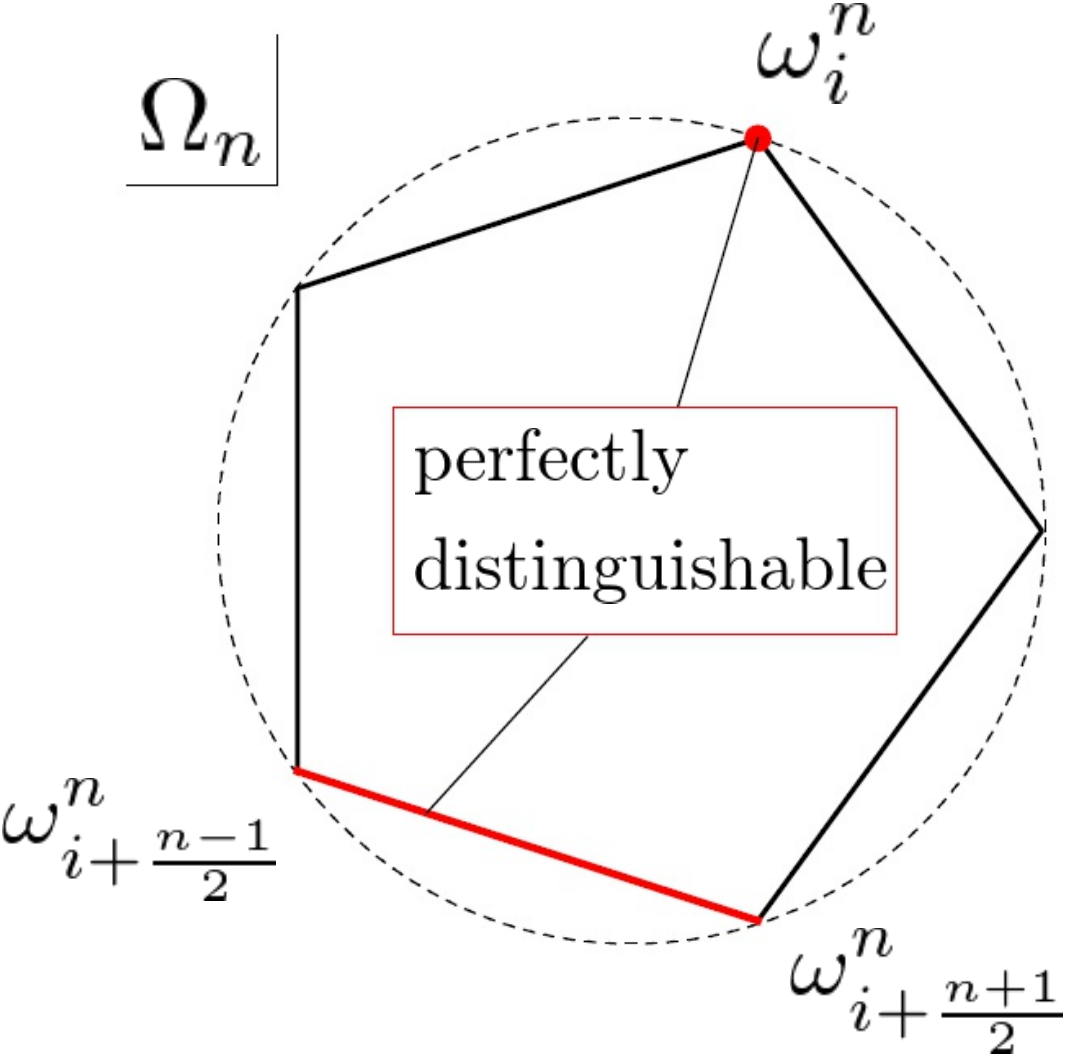}
		\subcaption{$n$ is an odd number.}
	\end{minipage}
	\hfill
	\begin{minipage}[b]{0.3\linewidth}
		\centering
		\includegraphics[scale=0.4]{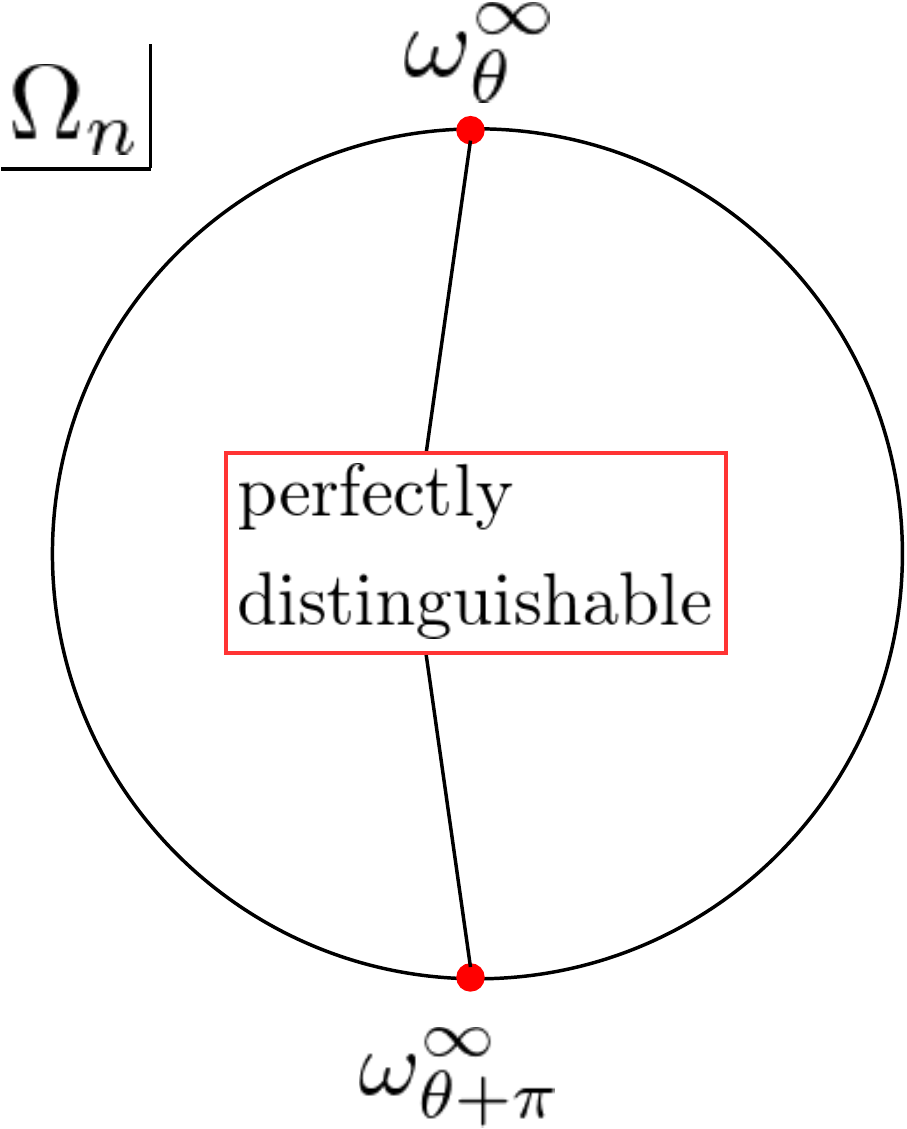}
		\subcaption{$n=\infty$.}
	\end{minipage}
	\hfill
	\null
	\caption{Pairs of perfectly distinguishable states in the $n$-gon state space.}
	\label{fig:n-gon}
\end{figure}
We first consider the case when $n$ is an even number greater than two. Calculating the inner products of pure effects and pure states, we obtain 
\begin{align*}
	e_{i}^{n}(\omega_{i}^{n})=e_{i}^{n}(\omega_{i-1}^{n})=1,\ \ e_{i}^{n}(\omega_{i+\frac{n}{2}-1}^{n})=e_{i}^{n}(\omega_{i+\frac{n}{2}}^{n})=0.
\end{align*}
These equations indicate that any state in $\Omega_{n}^{[i-1,\ i]}$ is perfectly distinguishable from any state in $\Omega_{n}^{[i+\frac{n}{2}-1,\ i+\frac{n}{2}]}$, where we define
\begin{align*}
\Omega_{n}^{[k-1,\ k]}=\{\omega\in\Omega_{n}\mid \omega=p\omega_{k-1}^{n}+(1-p)\omega_{k}^{n},\ 0\le p\le1\},
\end{align*}
since the measurement $\{e_{i}^{n},\ u-e_{i}^{n}\}$ distinguishes perfectly those two states. For odd $n\ (\ge3)$, we obtain
\begin{align*}
e_{i}^{n}(\omega_{i}^{n})=1,\ \ e_{i}^{n}(\omega_{i+\frac{n-1}{2}}^{n})=e^{n}_{i}(\omega_{i+\frac{n+1}{2}}^{n})=0.
\end{align*}
Hence, $\omega_{i}^{n}$ and an arbitrary state in $\Omega_{n}^{[i+\frac{n-1}{2},\ i+\frac{n+1}{2}]}$ are perfectly distinguishable.
Finally, when $n=\infty$, 
\begin{align*}
e_{\theta}^{\infty}(\omega_{\theta}^{\infty})=1, \ \ e_{\theta}^{\infty}(\omega_{\theta+\pi}^{\infty})=0
\end{align*}
hold, so there is only one perfectly distinguishable state for each pure state (see figure \ref{fig:n-gon}).

\subsection{Entropy of mixing in the framework of GPTs}
\label{2c}
In this part, we consider the thermodynamical entropy of mixing in a system composed of ideal gases with different internal degrees of freedom described by a GPT. In thermodynamics, it is well known that a mixture of several classically distinct ideal gases, such like a mixture of ideal hydrogens and nitrogens, causes an increase of entropy. The amount of increase by the mixture can be calculated under the assumption of the existence of semipermeable membranes which distinguish perfectly those particles. We assume in a similar way that if the internal states $\omega_{1}, \omega_{2}, \cdots, \omega_{l}$ described by a GPT are perfectly distinguishable, then there exist semipermeable membranes which can identify completely a state among them without disturbing every $\omega_{j}$ $(j=1, 2, \cdots, l)$. 

We consider ideal gases in thermal equilibrium with its temperature $T$, volume $V$, and $N$ particles, and do not focus on the mechanical part of the particles in the following. All of these $N$ particles are in the same internal state $\omega=\sum_{i=1}^{l}p_{i}\omega_{i}$, where $\{\omega_{1}, \omega_{2}, \cdots, \omega_{l}\}$ is a perfectly distinguishable set of states, and $\forall i,\  p_{i}\ge0,\ \mbox{and}\  \sum_{i=1}^{l}p_{i}=1$, meaning that this system is composed of the mixture of $l$ different kinds of particles whose internal states are $\omega_{1}, \omega_{2}, \cdots, \omega_{l}$ with probability weights $\{p_{1}, p_{2}, \cdots, p_{l}\}$. We note again that in this paper, classical species of particles are also regarded as the internal states of them. In classical thermodynamics, thermodynamical entropy is calculated by constructing concrete thermodynamical operations such as isothermal or adiabatic quasistatic operations. We follow this doctrine of thermodynamics also in GPTs that thermodynamical entropy, especially thermodynamical entropy of mixing, should be operationally-derived quantity. In fact, as shown in \cite{1367-2630-19-4-043025}, our assumption of the existence of semipermeable membranes makes it possible to realize concrete thermodynamical operations to calculate the thermodynamical entropy of mixing of the system mentioned above in the same way as von Neumann did when the internal degrees of freedom were quantum \cite{von1955mathematical}. Strictly speaking, it has been demonstrated operationally in \cite{1367-2630-19-4-043025} that the thermodynamical entropy of mixing in the system is
\begin{equation}
	\label{eq1}
	S(\omega)=\sum_{i=1}^{l}p_{i}S(\omega_{i})-\sum_{i=1}^{l}p_{i}\log p_{i}\ ,
\end{equation}
where $S(\sigma)$ means the per-particle thermodynamical entropy of mixing in the system which consists of particles in the same state $\sigma$, and we set the Boltzmann constant $\kB=1$ (also $0\log0=0$). In the process of deriving $\eqref{eq1}$, the additivity and extensivity of the thermodynamical entropy, and the continuity of $S$ with respect to states are assumed. The latter one is needed in order to apply $\eqref{eq1}$ to arbitrary states with arbitrary probability weights, while the operational derivation of $\eqref{eq1}$ has been given only when each $p_{i}N$ is the number of particles in the state $\omega_{i}$ and thus each $p_{i}$ is rational. We impose additional assumption that the entropy of any pure state equals to zero, that is, $S(\sigma)=0$ whenever $\sigma$ is a pure state.

\section{Main result}
\label{sec3}
Our main result is in the following form.
\begin{thm*}
	Consider a system in thermal equilibrium composed of ideal gases whose internal states are all  described by one of the elements of the $n$-gon ($n\ge3$) state space $\Omega_{n}$. The (per-particle) thermodynamical entropy of mixing $S:\Omega_{n}\rightarrow\mathbb{R}$ satisfying $\eqref{eq1}$ exists if and only if $n=3\ \mbox{or}\ \infty$, that is, the state space is classical or quantum-like.
\end{thm*}
\begin{proof}
	For $n=3$, as stated in the previous section, any $\omega\in\Omega_{3}$ is decomposed uniquely into perfectly distinguishable pure states as $\omega=p\omega_{0}^{3}+q\omega_{1}^{3}+(1-p-q)\omega_{2}^{3}$, where $\omega_{i}^{3}\ (i=0, 1, 2)$ are the three pure states in $\Omega_{3}$ and $\{p, q, 1-p-q\}$ is a probability weight. In this settings, we define $S$ as 
\begin{align*}
S(\omega)=-p\log p-q\log q-(1-p-q)\log (1-p-q).
\end{align*}
	This $S$ gives the well-defined entropy satisfying \eqref{eq1}. Similarly, when $n=\infty$, any state has only one decomposition into perfectly distinguishable (pure) states except for the central state of $\Omega_{\infty}$ (the maximally mixed state). For states which are not maximally mixed, we define $S$ as 
\begin{align*}
S(\omega)=H(p),
\end{align*}
	where we decompose a non-maximally-mixed $\omega\in\Omega_{\infty}$ as $\omega=p\omega_{\theta}^{\infty}+(1-p)\omega_{\theta+\pi}^{\infty}\ (0\le p\le1)$ and $H(p)=-p\log p-(1-p)\log(1-p)$ is the {\it 1-bit Shannon entropy}. We can apply this $S$ to the maximally mixed state, for the probability weights do not depend on the way of decompositions and they are always $\{\frac{1}{2},\frac{1}{2}\}$. Therefore, we can define successfully the thermodynamical entropy $S$ which meets \eqref{eq1} for $n=3, \infty$. In the following, we prove the only if part.

	The case when $n=4$ was proved in \cite{1367-2630-19-4-043025}, so we only consider $n\ge5$. At first, we assume $n$ is an even number, and consider the state $\omega_{\mathrm{P}}$ represented in figure \ref{fig:proof_even}, 
	\begin{figure}[h]
		\centering
		\includegraphics[scale=0.43]{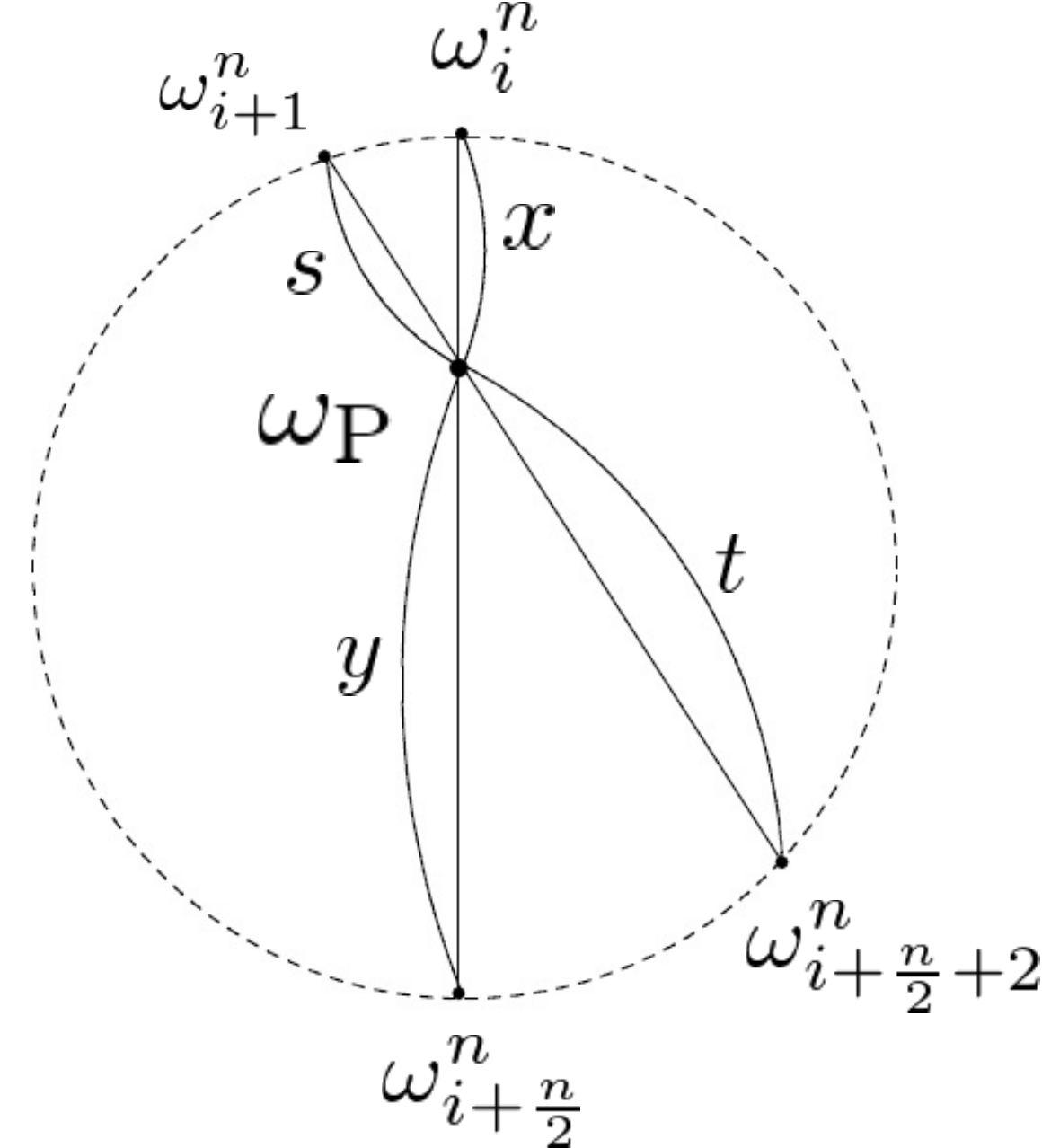}
		\caption{Illustration of the state $\omega_{\mathrm{P}}$.}
		\label{fig:proof_even}
	\end{figure}
	\\
	that is,
	\begin{align*}
		\omega_{\mathrm{P}}=\frac{y}{x+y}\omega_{i}^{n}+\frac{x}{x+y}\omega_{i+\frac{n}{2}}^{n}=\frac{t}{s+t}\omega_{i+1}^{n}+\frac{s}{s+t}\omega_{i+\frac{n}{2}+2}^{n},
	\end{align*}
	where $x, y, s, t$ are all nonnegative, and $x\le y$ and $s\le t$ as shown in figure \ref{fig:proof_even}. Note that $\{\omega_{i}^{n},\ \omega_{i+\frac{n}{2}}^{n}\}$ and $\{\omega_{i+1}^{n},\ \omega_{i+\frac{n}{2}+2}^{n}\}$ are two perfectly distinguishable pairs of pure states. From the observations in the previous section, we obtain two forms of the thermodynamical entropy of mixing:
	\begin{equation}
		\label{eq2}
		S(\omega_{\mathrm{P}})=H(\frac{x}{x+y})=H(\frac{s}{s+t}),
	\end{equation}
	which means
\begin{align*}
	\frac{x}{x+y}=\frac{s}{s+t}
\end{align*}
	because $x\le y$ and $s\le t$. On the other hand, by simple calculations (see appendix) we obtain
\begin{align*}
	\frac{s}{s+t}=\frac{x}{x+(\frac{\cos\frac{2\pi}{n}}{\cos\frac{\pi}{n}})^{2}\ y}.
\end{align*}
	It follows that  $(\frac{\cos\frac{2\pi}{n}}{\cos\frac{\pi}{n}})^{2}=1$ from these two equations, and because for even $n$, $(\frac{\cos\frac{2\pi}{n}}{\cos\frac{\pi}{n}})^{2}=1$ if and only if $n=\infty$ ($\cos(\frac{\pi}{n})=1$), the entropy \eqref{eq2} has been proved to be ill-defined.

	Next, we consider the case when $n$ is an odd number greater than three. We define the state $\omega_{\mathrm{A}}$ as $\omega_{\mathrm{A}}=\frac{1}{2}(\omega_{i+\frac{n-1}{2}}^{n}+\omega_{i+\frac{n+1}{2}}^{n})$,
	and consider two states $\omega_{\mathrm{Q}}$ and $\omega_{\mathrm{R}}$ shown in figure \ref{fig:proof_odd}, where $j=\frac{n+1}{4}\ \mbox{or}\ \frac{n-1}{4}$ corresponding to the case when $n\equiv3$ or $n\equiv1$ (mod 4) respectively.
\begin{figure}[!h]
	\hfill
	\begin{minipage}[b]{0.48\linewidth}
		\centering
		\includegraphics[scale=0.43]{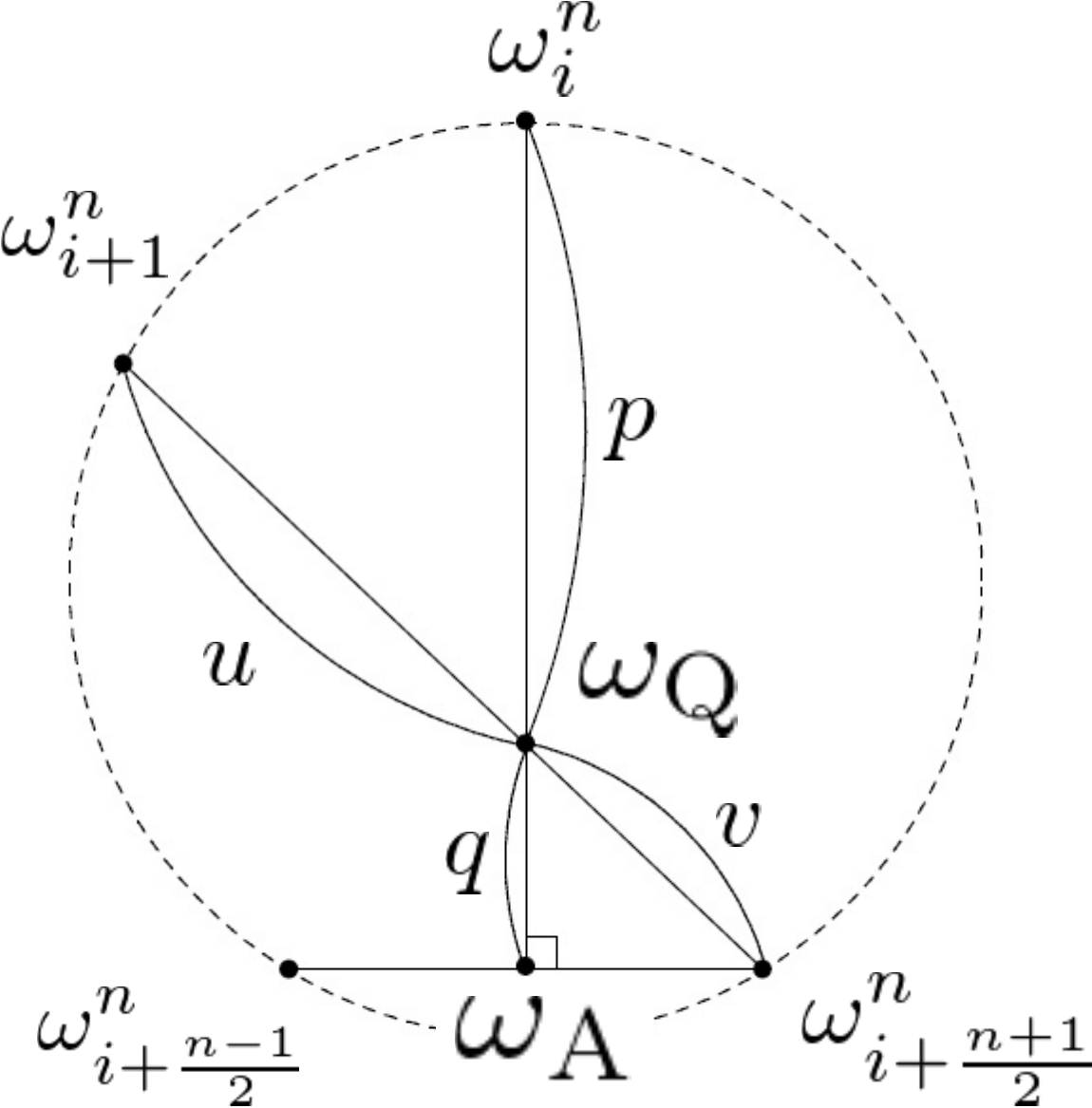}
		\subcaption{Illustration of the state $\omega_{\mathrm{Q}}$.}
		\label{fig:proof_odd1}
	\end{minipage}
	\hfill
	\begin{minipage}[b]{0.48\linewidth}
		\centering
		\includegraphics[scale=0.43]{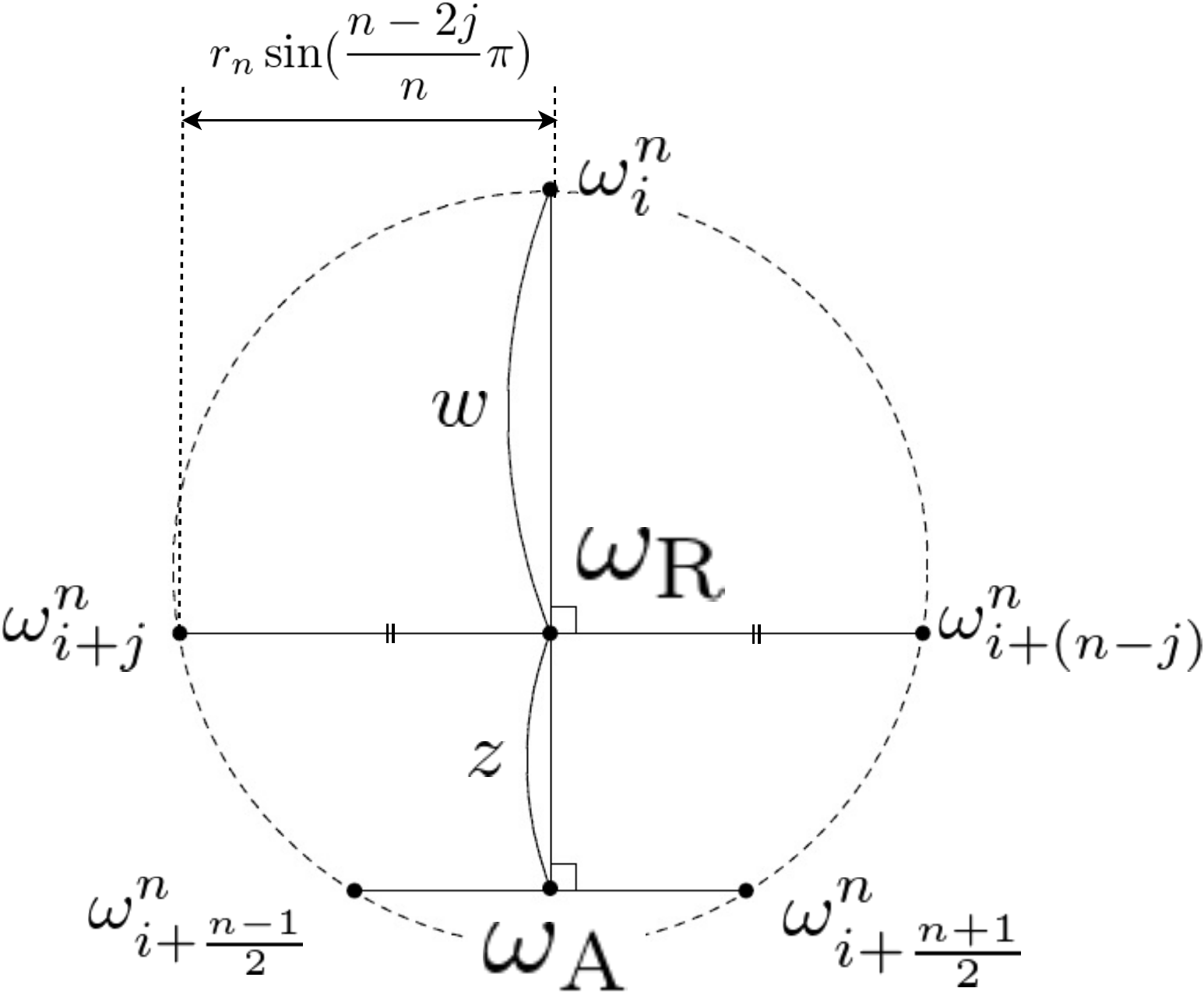}
		\subcaption{Illustration of the state $\omega_{\mathrm{R}}$.}
		\label{fig:proof_odd2}
	\end{minipage}
	\hfill
	\null
	\caption{Illustration of the states $\omega_{\mathrm{Q}}$ and $\omega_{\mathrm{R}}$.}
	\label{fig:proof_odd}
\end{figure}
Note that $\{\omega_{i}^{n},\ \omega_{\mathrm{A}}\}$ and $\{\omega_{i+1}^{n},\ \omega_{i+\frac{n+1}{2}}^{n}\}$ in figure \ref{fig:proof_odd1}, and $\{\omega_{i+j}^{n},\ \omega_{i+(n-j)}^{n}\}$ in figure \ref{fig:proof_odd2} are perfectly distinguishable pairs of states. Then, 
	\begin{align*}
		S(\omega_{\mathrm{Q}})=\frac{p}{p+q}S(\omega_{\mathrm{A}})+H(\frac{p}{p+q})=H(\frac{u}{u+v}),
	\end{align*}
	and
	\begin{align*}
		S(\omega_{\mathrm{R}})=\frac{w}{w+z}S(\omega_{\mathrm{A}})+H(\frac{w}{w+z})=H(\ \frac{1}{2}\ )
	\end{align*}
	hold. We assume that the entropies of the two states $\omega_{\mathrm{Q}},\ \omega_{\mathrm{R}}$ are well-defined (so is $\omega_{\mathrm{A}}$). Then,
	\begin{align}
		S(\omega_{\mathrm{A}})&=\frac{p+q}{p}\{H(\frac{u}{u+v})-H(\frac{p}{p+q})\}\label{eq:3} \\
		&=\frac{w+z}{w}\{H(\frac{1}{2})-H(\frac{w}{w+z})\}\label{eq:4}
	\end{align}
	holds. By elementary geometrical calculations (see appendix) and letting $\alpha=\sin\frac{\pi}{2n}$, we obtain
	\begin{align*}
&\frac{p+q}{p}\{H(\frac{u}{u+v})-H(\frac{p}{p+q})\}\\
&\qquad\qquad\qquad=2\alpha^{2}\log2+\frac{1-4\alpha^{2}}{2}\log(1-4\alpha^{2})-(1-2\alpha^{2})\log(1-2\alpha^{2})
\end{align*}
and
\begin{align*}
\frac{w+z}{w}\{H(\frac{1}{2})-H(\frac{w}{w+z})\}=(1\mp2\alpha)\log(1\mp2\alpha)-(2\mp2\alpha)\log(1\mp\alpha),
\end{align*}
that is, 
\begin{align}
		S(\omega_{\mathrm{A}})&=2\alpha^{2}\log2+\frac{1-4\alpha^{2}}{2}\log(1-4\alpha^{2})-(1-2\alpha^{2})\log(1-2\alpha^{2})\label{eq:5}\\
		&=(1\mp2\alpha)\log(1\mp2\alpha)-(2\mp2\alpha)\log(1\mp\alpha)\label{eq:6},
	\end{align}
	where the upper and lower signs correspond to the case of $n\equiv3$ and $n\equiv1$ (mod 4) respectively.
	\begin{figure}[h]
		\hfill
		\begin{minipage}[b]{0.48\linewidth}
			\includegraphics[scale=0.41]{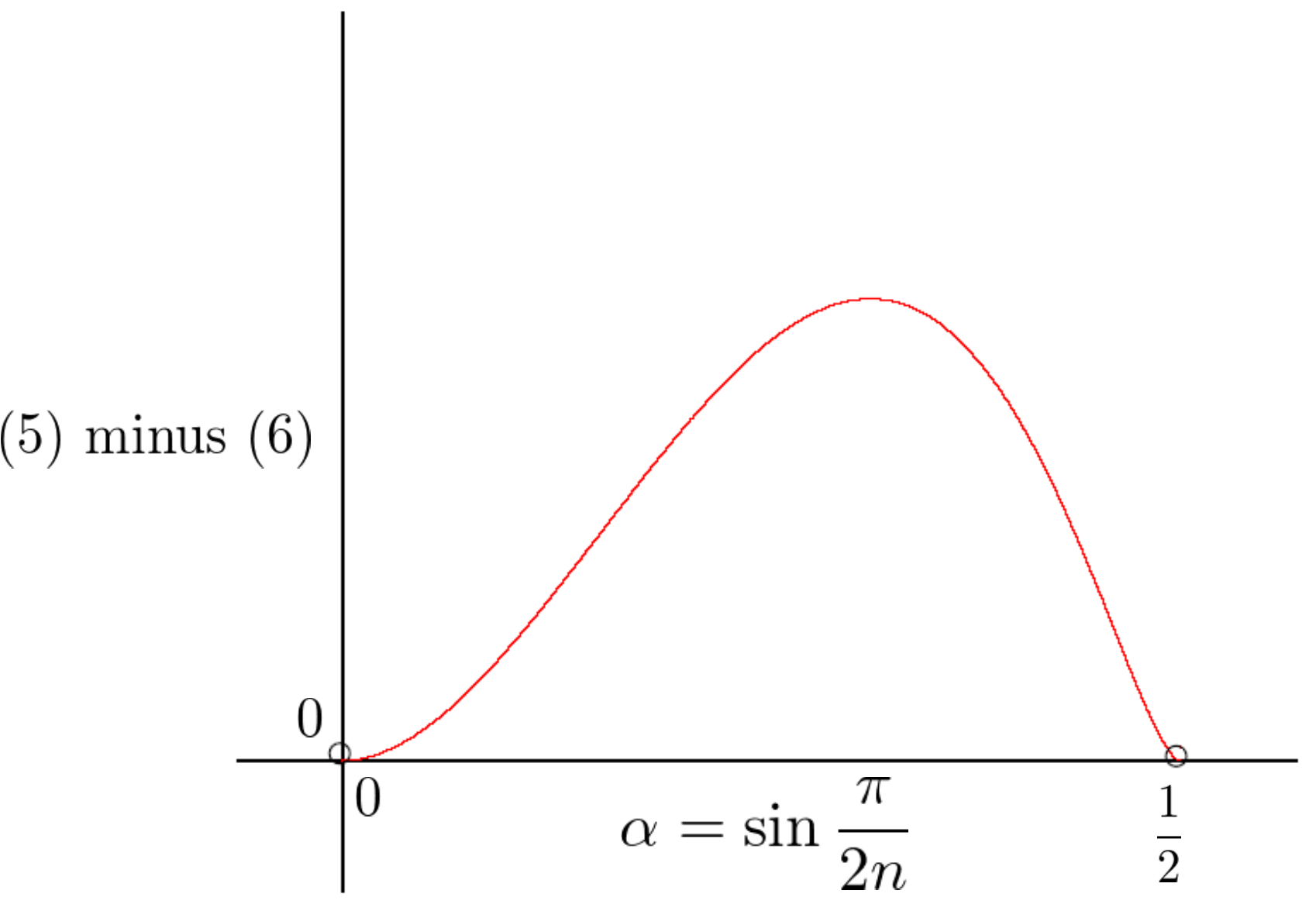}
			\subcaption{$n\equiv3$.}
			\label{fig:graph_equiv3}
		\end{minipage}
		\hfill
		\begin{minipage}[b]{0.48\linewidth}
			\centering
			\includegraphics[scale=0.4]{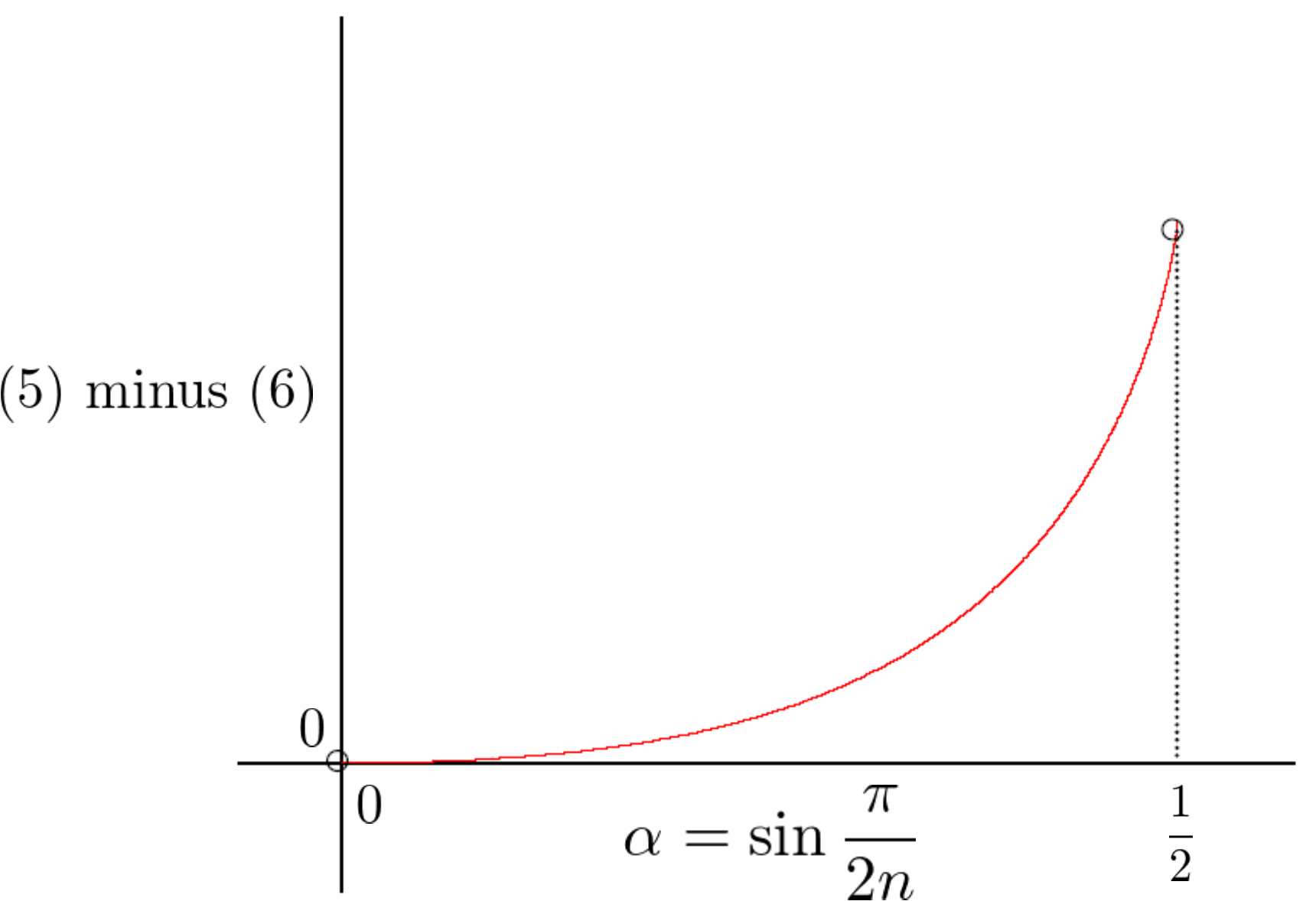}
			\subcaption{$n\equiv1$.}
			\label{fig:graph_equiv1}
		\end{minipage}
		\hfill\null
		\caption{The difference of the two values \eqref{eq:5} and \eqref{eq:6} of $S(\omega_{\mathrm{A}})$.}
		\label{fig:graph}
	\end{figure}
	The differences between \eqref{eq:5} and \eqref{eq:6} in the case of $n\equiv3$ and $n\equiv1$ are displayed in	figure \ref{fig:graph_equiv3} and \ref{fig:graph_equiv1} respectively, and we can see that the two forms of $S(\omega_{\mathrm{A}})$ shown in \eqref{eq:5} and \eqref{eq:6} do not agree with each other. In conclusion, it has been proved that if $n\neq3,\ \infty$, then there exists some state whose thermodynamical entropy of mixing is ill-defined.
\end{proof}

\section{Conclusions and future prospects}
\label{sec4}
Overall, although we have only considered the $n$-gon theories embedded in $\mathbb{R}^{3}$, we showed that only classical and quantum-like theories allowed the entropy \eqref{eq1} to be consistent. What makes our claim more reasonable is that we can see the ill-defined values of entropy become well-defined if $n=3,\ \infty$ in our proof. For example, when $n$ is an odd number, $\alpha=\sin\frac{\pi}{2n}$ equals to $\frac{1}{2}$ or $0$ if $n$ equals to three or infinite, respectively, and two values \eqref{eq:5} and \eqref{eq:6} coincide with each other in these cases (see fig. \ref{fig:graph}).

Note that similar results were obtained in \cite{1367-2630-19-4-043025}, where it was assumed that any state could be represented as a convex combination of perfectly distinguishable pure states. However, a state of the $n$-gon theory is not always represented by a convex combination of perfectly distinguishable pure states. For instance, we can see from fig. \ref{fig:n-gon} that the state $\omega_{\mathrm{A}}$ in fig. \ref{fig:proof_odd1} or fig. \ref{fig:proof_odd2} can not be decomposed into perfectly distinguishable pure states. Thus, the regular polygon theories generally do not satisfy the assumption in the previous study \cite{1367-2630-19-4-043025}, and our result is the one about the exsistence of well-defined thermodynamical entropy in such a broader class of theories where ``spectral decompositions" of states are not generally possible.

obtained in such a broader class of theories.

Further research is required to reveal if we can obtain the same results in higher dimensional cases (in GPTs, higher dimensional classical theories are known to be described generally by simplexes, but higher dimensional quantum theories have more complicated structures \cite{KIMURA2003339, Bengtsson2013}). Moreover, the proof of our main theorem indicates that the entropy discussed above is defined successfully in other theories where the probability coefficients obtained when a state is decomposed into perfectly distinguishable states are unique even though the state space is neither classical nor quantum. This means that we need to impose additional conditions on the entropy to remove those``unreasonable" theories, which is also a future problem.

\section*{Acknowledgement}
The author would like to thank Takayuki Miyadera (Kyoto University) for many helpful comments, and also wishes to thank anonymous referees for valuable remarks.

\appendix
\section{Detailed proof of the theorem}
\label{appendix}
In this appendix, we illustrate how to derive
\begin{equation}
\label{eq:app1}
\frac{s}{s+t}=\frac{x}{x+(\frac{\cos\frac{2\pi}{n}}{\cos\frac{\pi}{n}})^{2}\ y},
\end{equation}
and
\begin{equation}
\label{eq:app2}
\begin{aligned}
S(\omega_{\mathrm{A}})&=2\alpha^{2}\log2+\frac{1-4\alpha^{2}}{2}\log(1-4\alpha^{2})-(1-2\alpha^{2})\log(1-2\alpha^{2})\\
&=(1\mp2\alpha)\log(1\mp2\alpha)-(2\mp2\alpha)\log(1\mp\alpha)
\end{aligned}
\end{equation}
in the proof of the main theorem (remember that $\alpha=\sin{\frac{\pi}{2n}}$).

To derive the former, we apply sine theorem to figure \ref{fig:proof_even_app}, and thus obtain
\begin{figure}[h]
	\hfill
	\begin{minipage}[b]{0.45\linewidth}
		\centering
		\includegraphics[scale=0.43]{proof_even1.pdf}
	\end{minipage}
	\hfill
	\begin{minipage}[b]{0.45\linewidth}
		\centering
		\includegraphics[scale=0.43]{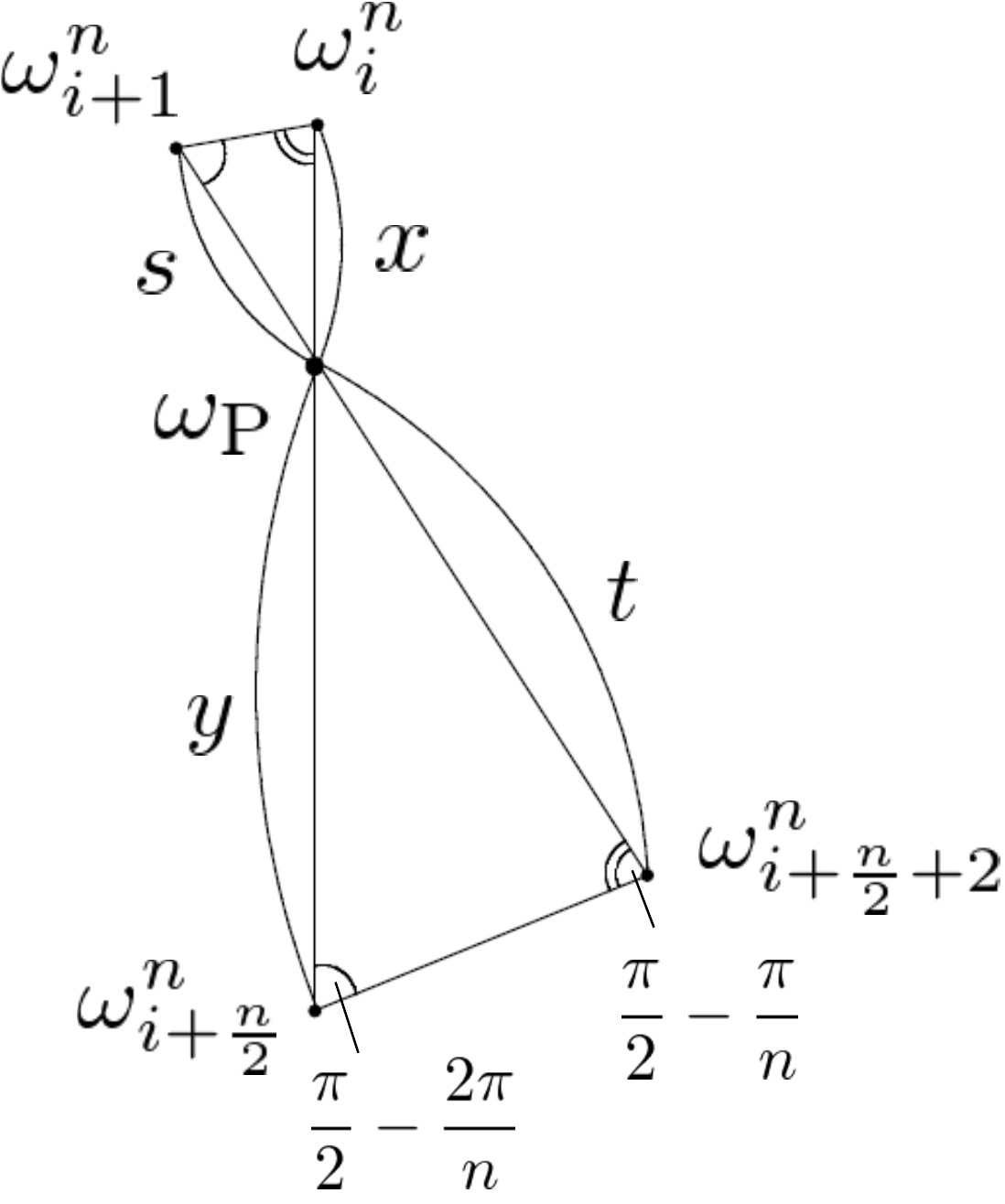}
	\end{minipage}
	\hfill\null
	\caption{}
	\label{fig:proof_even_app}
\end{figure}
\begin{align*}
\frac{x}{\sin(\frac{\pi}{2}-\frac{2\pi}{n})}=\frac{s}{\sin(\frac{\pi}{2}-\frac{\pi}{n})},\ \ 
\frac{y}{\sin(\frac{\pi}{2}-\frac{\pi}{n})}=\frac{t}{\sin(\frac{\pi}{2}-\frac{2\pi}{n})},
\end{align*}
namely
\begin{align*}
\frac{s}{s+t}=\frac{x}{x+(\frac{\cos\frac{2\pi}{n}}{\cos\frac{\pi}{n}})^{2}\ y}.
\end{align*}

For the latter, from figure \ref{fig:proof_odd1_app}, we obtain
\begin{figure}[h]
	\hfill
	\begin{minipage}[b]{0.45\linewidth}
		\centering
		\includegraphics[scale=0.43]{proof_odd1.pdf}
	\end{minipage}
	\hfill
	\begin{minipage}[b]{0.45\linewidth}
		\centering
		\includegraphics[scale=0.44]{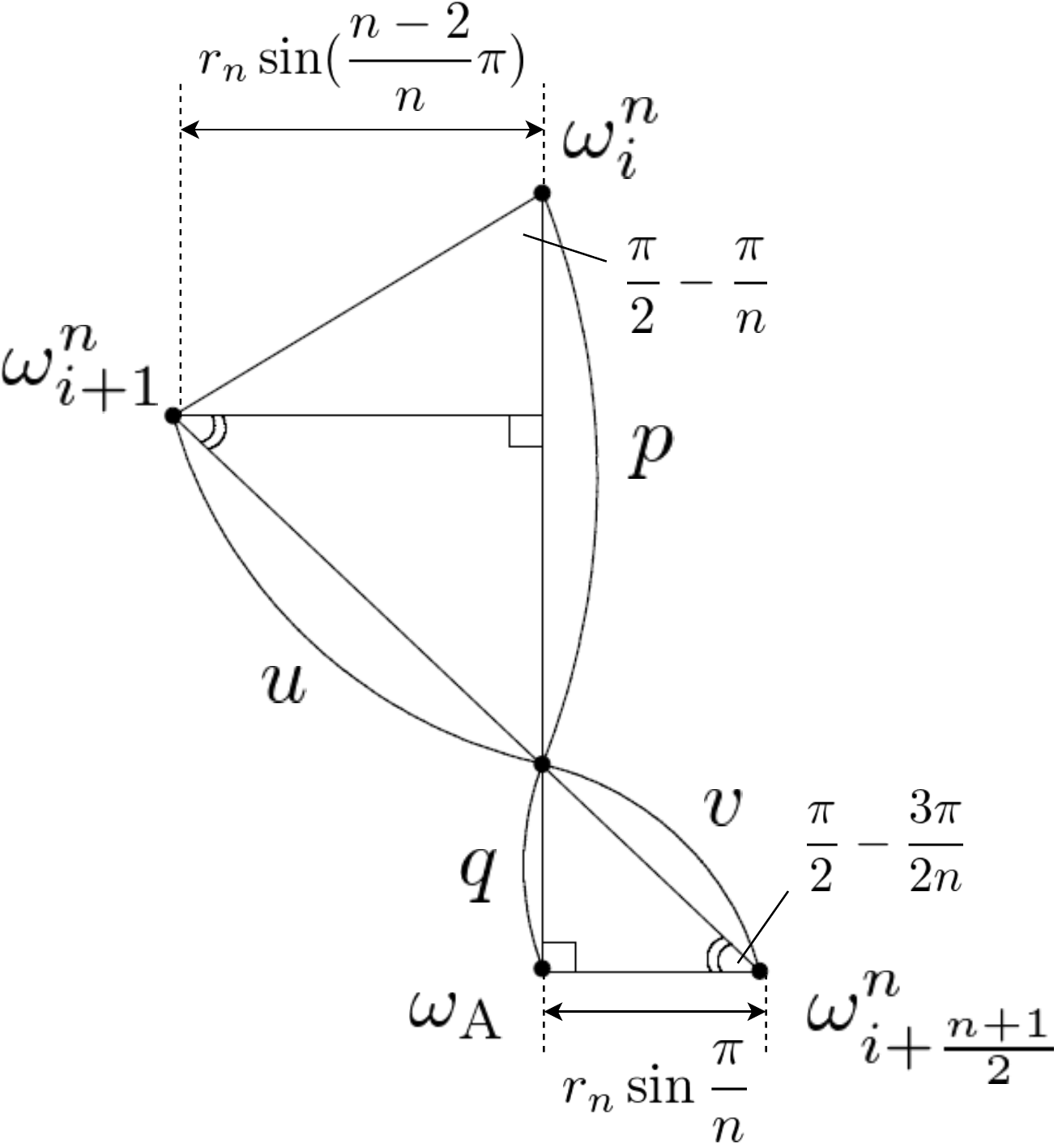}
	\end{minipage}
	\hfill\null
	\caption{}
	\label{fig:proof_odd1_app}
\end{figure}
\begin{align*}
\frac{u}{v}=\frac{r_{n}\sin(\frac{n-2}{n}\pi)}{r_{n}\sin\frac{\pi}{n}}
=\frac{\sin\frac{2\pi}{n}}{\sin\frac{\pi}{n}}=2\cos\frac{\pi}{n},
\end{align*}
and by sine theorem, 
\begin{align*}
\frac{p}{\sin(\frac{\pi}{2}-\frac{\pi}{2n})}=\frac{u}{\sin(\frac{\pi}{2}-\frac{\pi}{n})},\ \ 
q=v\cos\frac{3\pi}{2n}
\end{align*}
hold. Therefore,
\begin{align*}
\frac{q}{p}&=\frac{v}{u}\cdot\frac{\cos\frac{\pi}{n}\cdot\cos\frac{3\pi}{2n}}{\cos\frac{\pi}{2n}}\\
&=\frac{\cos\frac{3\pi}{2n}}{2\cos\frac{\pi}{2n}}\\
&=\frac{1}{2}(4\cos^{2}\frac{\pi}{2n}-3)\\
&=\frac{1}{2}(2\cos\frac{\pi}{n}-1).
\end{align*}
\begin{figure}[h]
	\centering
	\includegraphics[scale=0.43]{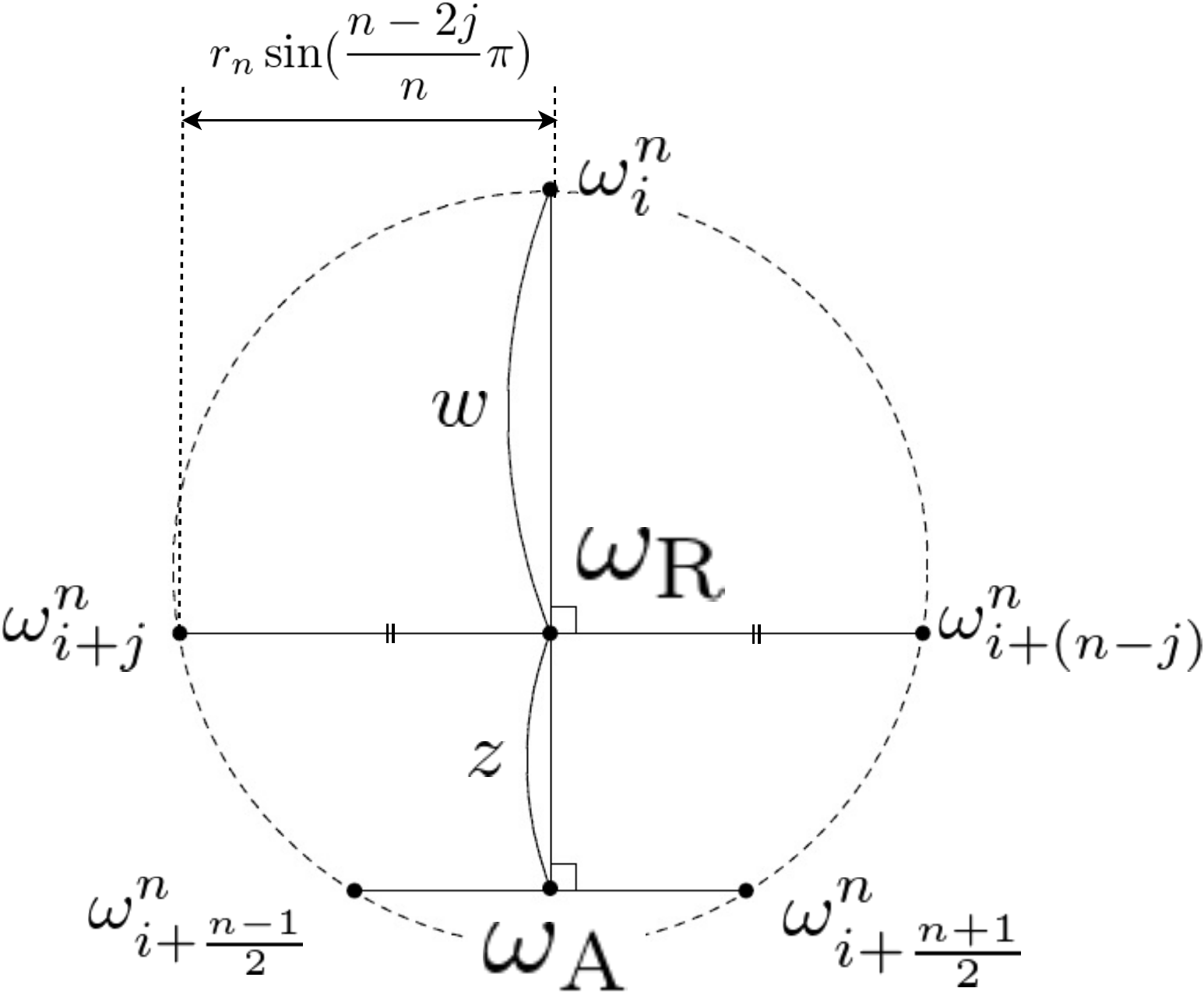}
	\caption{}
	\label{fig:proof_odd2_app}
\end{figure}
On the other hand, from figure \ref{fig:proof_odd2_app}, we obtain
\begin{align}
\label{eq:app3}
\begin{aligned}
\frac{z}{w}&=\frac{2r_{n}\cos^{2}\frac{\pi}{2n}-2r_{n}\cos^{2}(\frac{n-2j}{2n}\pi)}{2r_{n}\cos^{2}(\frac{n-2j}{2n}\pi)}\\
&=\frac{\cos^{2}\frac{\pi}{2n}-\sin^{2}\frac{j\pi}{n}}{\sin^{2}\frac{j\pi}{n}}\ \ \ \ \mbox{for}\ \ j=\frac{n\pm1}{4}.
\end{aligned}
\end{align}
Since
\begin{equation}
\label{eq:app4}
\begin{aligned}
\sin^{2}\frac{j\pi}{n}&=\frac{1}{2}(1-\cos\frac{2j\pi}{n})\\
&=\frac{1}{2}(1-\cos\frac{(n\pm1)\pi}{2n})\\
&=\frac{1}{2}(1\pm\sin\frac{\pi}{2n}),
\end{aligned}
\end{equation}
the equation above can be written as
\begin{align*}
\frac{z}{w}&=\frac{2(1+\sin\frac{\pi}{2n})(1-\sin\frac{\pi}{2n})-(1\pm\sin\frac{\pi}{2n})}{1\pm\sin\frac{\pi}{2n}}\\
&=1\mp2\sin\frac{\pi}{2n},
\end{align*}
where the double sign corresponds to the ones in \eqref{eq:app3} and \eqref{eq:app4}, and the upper and lower sign correspond to the case of $n\equiv3$ and $n\equiv1$ (mod 4) respectively. Substituting these results to \eqref{eq:3} and \eqref{eq:4}, we obtain
\begin{align*}
S(\omega_{\mathrm{A}})
&=(\frac{2\cos\frac{\pi}{n}+1}{2})\{H(\frac{1}{2\cos\frac{\pi}{n}+1})-H(\frac{2}{2\cos\frac{\pi}{n}+1})\} \\
&=(2\mp2\sin\frac{\pi}{2n})\{H(\frac{1}{2})-H(\frac{1}{2\mp2\sin\frac{\pi}{2n}})\},
\end{align*}
which means \eqref{eq:app2}.

\bibliographystyle{unsrt} 
\bibliography{ref}

\end{document}